\documentclass[11pt]{eptcs}
\providecommand{\event}{}

\providecommand{\firstpage}{1}
\providecommand{\urlalt}[2]{\href{#1}{#2}}
\providecommand{\doi}[1]{doi:\urlalt{http://dx.doi.org/#1}{#1}}
\usepackage{amsfonts,amsthm}
\usepackage{amssymb,amsmath}
\usepackage{cite}
\usepackage{graphicx}
\usepackage[english]{babel}
\usepackage{ifpdf,multicol,breakurl}
\usepackage{xcolor}
\usepackage[all]{xy}  
\xyoption{arc}
\xyoption{poly}                    
\usepackage{./Notes}
\newcommand{\disc}[1]{\textsf{Prob}(#1)}
\newcommand{\discs}[1]{\textsf{Prob$_S$}(#1)}

\newcommand{\discr}{\textsf{Prob}}

\newcommand{\Capp}[1]{\textsf{Cap}(#1)}

\newcommand{\X}{{\mathcal X}}
\newcommand{\I}{{\mathcal I}}
\newcommand{\Y}{{\mathcal Y}}
\newcommand{\K}{\mathcal{K}}
\newcommand{\M}{{\mathcal M}}

\renewcommand{\Re}{{\mathbb R}} 
 
\newcommand{\da}{\mathop{\downarrow}}
\newcommand{\ua}{\mathop{\uparrow}}

\newcommand{\Nat}{{\mathbb N}}
\newcommand{\ST}{\textsf{ST}}
\newcommand{\DT}{\textsf{DT}}

\newcommand{\Aff}{\textsf{Aff}}
\newcommand\conv{\overline{\mbox{\rm conv}}} 
\def\defi{\buildrel\rm def 
\over=}                    
\def\titlerunning{Probabilistic Monads, Domains and Classical Information}
\def\authorrunning{Mislove}
\begin{document}
\title{Probabilistic Monads, Domains and Classical Information}
\author{Michael Mislove\footnote{The support of the US Office of Naval Research is gratefully acknowledged.}
\institute{Tulane University\\New Orleans, LA 70118}
}

\maketitle

\begin{abstract}
Shannon's classical information theory \cite{shannon} uses probability theory to analyze channels as mechanisms for  information flow. In this paper, we generalize results from \cite{martin2} for binary channels to show how some more modern tools --- probabilistic monads and domain theory in particular --- can be used to model classical channels. As initiated in \cite{martin2}, the point of departure is to consider the family of channels with fixed inputs and outputs, rather than trying to analyze channels one at a time. The results show that domain theory has a role to play in the capacity of channels; in particular, the $n\times n$-stochastic matrices, which are the classical channels having the same sized input as output, admit a quotient compact ordered space which is a domain, and the capacity map factors through this quotient via a Scott-continuous map that measures the quotient domain.  We also comment on how some of our results relate to recent discoveries about quantum channels and free affine monoids. 
\end{abstract}

\section{Introduction}
Classical information theory has its foundations in the seminal work of Claude Shannon~\cite{shannon}, who first conceived of analyzing the behavior of channels using \emph{entropy} and deriving a formula for \emph{channel capacity} based on mutual information (cf.~\cite{cover} for a modern presentation of the basic results). Recent work of Martin, et al.~\cite{martin2} reveals that the theory of compact, affine monoids and domain theory can be used to analyze the family of binary  channels. In this paper, our goal is to generalize the results in \cite{martin2} to the case of $n\times n$-channels --- channels that have $n$ input ports and $n$ output ports. Our approach also uses the monadic properties of probability distributions to give an abstract presentation of how channels arise, and that clarifies the role of the doubly stochastic matrices, which are special channels.   While our work focuses on the classical case, the situation around quantum information and quantum channels is also a concern, and we point out how our results relate to some recent work \cite{crowd,MartinFeng} on quantum qubit channels and free affine monoids. While most of the ingredients we piece together are not new, we believe the approach we present does represent a new way in which to understand families of channels and some of their important features. 

The rest of the paper is structured as follows. In the next section, we describe three monads based on the probability measures over compact spaces, compact monoids and compact groups. Each of these is used to present some aspect of the classical channels. We then introduce topology, and show how the capacity of a channel can be viewed from a topological perspective. The main result here is that capacity is the maximum distance from the surface determined by the entropy function and the underlying polytope generated by the rows of a channel matrix, viewed as vectors in $\Re^n$ for appropriate $n$. This leads to a generalization of Jensen's Lemma that characterizes strictly concave functions. Domain theory is then introduced, as applied to the finitely-generated polytopes residing in a compact convex set, ordered by reverse inclusion. Here we characterize when proper maps measure a domain, in the sense of Martin~\cite{martin:thesis}; a closely related result can be found in \cite{panangad}. Finally, we return to the compact monoid of $n\times n$-stochastic matrices and show that it has a natural, algebraically-defined pre-order relative to which capacity measures the quotient partial order, which is a compact ordered space. The capacity mapping is also shown to be strictly monotone with respect to this pre-order, which means that strictly smaller channel matrices have strictly smaller capacity. We close with a summary and comments about future work. 

\section{Three probabilistic monads}
The categorical presentation of classical information relies on three monads, each of which has the family $\disc{X}$ of probability distributions over a set $X$ as the object-level of the left adjoint. The first of these starts with compact Hausdorff spaces, and uses several results from functional analysis: standard references for this material are \cite{bour,rudin}. We present these monads in turn:
\subsection{A spatial monad}
We begin with the probability measure monad over topological spaces. If $X$ is a compact Hausdorff space, then $C(X,\Re)$, the family of continuous, real-valued functions defined on $X$, is a  Banach space (complete, normed linear space) in the $\sup$-norm. 
The Banach space dual of $C(X,\Re)$, denoted $C(X,\Re)^*$ consists of all continuous linear functionals from $C(X,\Re)$ into $\Re$. $C(X,\Re)^*$ is another Banach space, and the Riesz Representation Theorem implies this is the Banach space of Radon measures on $X$ (those that are both inner- and outer regular). The unit sphere of $C(X,\Re)^*$ is the family $\disc{X}$ of probability measures over $X$. If we endow $\disc{X}$ with the weak$^*$ topology (the weakest topology making all continuous linear functionals into $\Re$ continuous), then $\disc{X}$ becomes a compact, Hausdorff space, by the Banach-Alaoglu Theorem. $\discr$ extends to a functor $\discr_S\colon \textsf{Comp}\to \textsf{CompConv}_{LC}$ from  the category of compact Hausdorff spaces and continuous maps, to the category of compact, convex,  locally convex spaces  and continuous affine maps, via $\discs{X} = \disc{X}$ and $f\colon X\to Y$ maps to $\discs{f}\colon \discs{X}\to\discs{Y}$ by $\discs{f}(\mu)(A) = \mu(f^{-1}(A))$, for each Borel set $A\subseteq Y$.

 Moreover, if the mapping $x\mapsto \delta_x\colon  X\to C(X,\Re)^*$ sending a point to the Dirac measure it defines, is a continuous mapping into the weak$^*$ topology. Since $X$ is compact Hausdorff, Urysohn's Lemma implies $C(X,\Re)$ separates the points of $X$, and so $x\mapsto \delta_x$ is a homeomorphism onto its image. Another application of Urysohn's Lemma shows each Dirac measure is an extreme point $\disc{X}$ and in fact the Dirac measures form the set of extreme points of $\disc{X}$.

A \emph{simple measure} is a finite, convex combination of Dirac measures, i.e., one of the form $\sum_{i\leq n} r_i \delta_{x_i}$, where $r_i\geq 0$, $\sum_i r_i = 1$, and $x_i\in X$ for each $i$. We let $\textsf{Prob}_{sim}(X)$ denote this family. 
 The Krein-Milman Theorem implies that $\textsf{Prob}_{sim}(X)$ is weak$^*$ dense among the probability measures. So, if $f\colon X\to C$ is a continuous function from $X$ into a compact subset of a locally convex vector space, then the function $\widehat{f}(\delta_x) = f(x)$ extends uniquely to continuous function $\widehat{f}(\sum_{i\leq n} r_i \delta_{x_i}) = \sum_{i\leq n} r_if(x_i)$, and then to all of $\disc{X}$, by the density of the simple measures. Obviously, $\widehat{f}(\delta_x) = f(x)$. 
 
We conclude that the functor $\discr_S$ is left adjoint to the forgetful functor. In fact, \discr$_S$ defines a monad,  where the unit of the adjunction is the mapping $\eta_X(x) = \delta_x$ and the multiplication $\mu\colon \disc{\disc{X}} \to \disc{X}$ is integration. 

\begin{theorem}
The functor $\discr_S$ sending a compact space to its family of probability measure in the weak$^*$ topology defines a monad on the category \textsf{Comp}. The unit of the monad sends a point $x\in X$ to the Dirac measure $\delta_x$, and the image of the unit is the set of extreme points in $\disc{X}$.
\end{theorem}

\begin{definition} Let $X$ and $Y$ be compact Hausdorff spaces. A \emph{(lossless) noisy channel from $X$ to $Y$} is a mapping $f\colon X\to \disc{Y}$. 
\end{definition}
Since \discr$_S$ is a monad, each channel $f\colon X\to \disc{Y}$ corresponds uniquely to a continuous, affine mapping $\discs{f}\colon\disc{X}\to \disc{Y}$ in the Kleisli category $\K_{\discr_S}$ of \discr$_S$. 
\begin{example}
Let $n \geq 1$ and let $\overline{n} = \{0,\ldots,n-1\}$ be the discrete, compact space. Then $\disc{\overline{n}}$ is the family of probability distributions on $n$ points, and given $m\geq 1$, a channel $f\colon \overline{m}\to\disc{\overline{n}}$ is an $m\times n$-stochastic matrix. 
The family $\ST(m,n)$ of $m\times n$-stochastic matrices is then the family of lossless, noisy channels from $\overline{m}$ to $\overline{n}$. 
Moreover, from our comment about the Kleisli category $\K_{\discr_S}$, we conclude that the family of morphisms $\K_{\discr_S}(\overline{m},\disc{\overline{n}})$ is $\ST(m,n)  \hookrightarrow \Aff(\disc{\overline{m}},\disc{\overline{n}})$, where $\Aff(\disc{\overline{m}},\disc{\overline{n}})$ is the family of continuous affine maps from $\disc{\overline{m}}$ to $\disc{\overline{n}}$.   
\end{example}
This first probabilistic monad shows that classical channels correspond to mappings in the Kleisli category of the ``spatial" monad \discr$_S$ on the category \textsf{Comp}. If we let $m = n$, then $\ST(n)\ {\buildrel \mathrm{def}\over =}\ \ST(n,n)$ is also a monoid 
using composition in the Kleisli category: if $f,g\in\ST(n)$, and $\widehat{g}\colon \disc{\overline{n}}\to \disc{\overline{n}}$ is the extension of $g$, then $g\circ f ::= \widehat{g}\circ f\in\ST(n)$.
We next present a second monad that gives another account of this special case.

\subsection{A monad on monoids}
The second monad we define is based on the category \textsf{CMon} of compact monoids and compact monoid homomorphisms. More precisely, a \emph{compact monoid} is a monoid $S$ --- a non-empty set endowed with an associative binary operation $(x,y)\mapsto xy\colon S\times S \to S$ that also has an identity element, $1_S$ ---  that also is a compact Hausdorff space for which the multiplication is continuous.  We can apply the probability functor to such an $S$ to obtain the compact convex (Hausdorff) space $\disc{S}$ of probability measures on $S$. If we denote multiplication on $S$ by $\cdot$, then $\disc{\cdot}\colon \disc{S\times S}\to \disc{S}$, and since $\iota_S\colon \disc{S}\times \disc{S}\hookrightarrow \disc{S\times S}$ is an embedding, we have a continuous affine map $\disc{\cdot}\circ \iota_S\colon \disc{S}\times \disc{S}\to \disc{S}$. This map is called \emph{convolution}, and we denote $(\disc{\cdot}\circ\iota_S)(\mu,\nu) = \mu\ast_S\nu$. It is routine to show convolution is associative, so $\disc{S}$ is a compact affine monoid. 

If $\phi\colon S\to T$ is a morphism of compact semigroups, 
then $\discr(\phi)\colon\discr(S)\to\discr(T)$ is defined by $\discr(\phi)(\mu)(f) = \int f\circ \phi\, d\mu$ for any $f\colon T\to \Re$. 
If $\mu,\nu\in\discr(S)$ 
and $f\in C(T,\Re)$, then
\begin{eqnarray*}
\discr(\mu\ast_S\nu)(f) = \int_S (f\circ\phi)\, d(\mu\ast_S\nu) &=& \int_S \int_S f\circ\phi\circ m_S\, d\mu d\nu \\
&{\buildrel 1\over =}& \int_S \int_S f\circ m_T\circ (\phi\times \phi)\, d\mu d\nu\\ 
&=& (\discr(\phi)(\mu)\ast_T\discr(\phi)(\nu))(f),
\end{eqnarray*}
where $m_S\colon S\times S\to S$, $m_T\colon T\times T\to T$ are the semigroup operations, $\ast_S, \ast_T$ denote convolution, and where ${\buildrel 1\over =}$ follows from the fact that $\phi$ is a homomorphism.  Thus $\discr_S\colon (\discr(S),\ast_S)\to (\discr(T),\ast_T)$ is a semigroup homomorphism. Finally, the fact that $\discr(\phi)$ preserves the identity follows from the observation that $\delta_x\ast_S\delta_y = \delta_{xy}$, which implies that $\delta_1$ is an identity for the simple measures, and consequently for all measures since the simple measures are dense. 

It follows that restricting $\discr_S$ to the subcategory \textsf{CMon} of \textsf{Comp} yields a functor $\discr_M\colon \textsf{CMon}\to \textsf{CAM}$ into the category of  locally convex compact affine monoids and continuous affine monoid maps. 
\begin{theorem}\label{thm:monoid}
The restriction of \discr$_S$ to \textsf{Mon} induces a monad \discr$_M$ whose target is \textsf{CAM}, the category of locally convex compact affine monoids and continuous, affine monoid homomorphisms. The unit of the monad is again the Dirac map, and its image is again the set of extreme points of $\disc{S}$.
\end{theorem}
Applying the same reasoning as for \discr$_S$, we see that each continuous monoid homomorphism $\phi\colon S\to \disc{T}$ corresponds to a unique morphism of compact affine monoids, $\widehat{\phi}\colon\disc{S}\to \disc{T}$ in the Kleisli category $\K_{\discr_M}$. However, in relation to classical channels, our interest is in the object level of \discr$_M$:
\begin{example}
We return to the example \ST$(n)$ of stochastic $n\times n$-matrices. These arise as channels on a discrete, $n$-element set. For such a set $\overline{n}$, the selfmaps of $\overline{n}$ form a finite --- hence compact --- monoid. If we denote this monoid by $[\overline{n}\to \overline{n}]$, then applying $\discr_M$ we obtain a compact affine monoid $\discr_M([\overline{n}\to\overline{n}])$. 

Now, $[\overline{n}\to\overline{n}]\hookrightarrow [\overline{n}\to \disc{\overline{n}}]$ by $f\mapsto \eta_{\overline{n}}\circ f$, where $\eta_{\overline{n}}$ is the unit for \discr$_S$. Since $\ST(n) = [\overline{n}\to\disc{\overline{n}}]$ is a compact affine monoid, this mapping extends to a morphism of compact affine monoids 
$$\sum_{i\leq k} r_i \delta_{f_i}\mapsto \sum_{i\leq k} r_i \eta_{\overline{n}}\circ f_i\colon \discr_M([\overline{n},\overline{n}])\to \ST(n).$$
Since $\{ \eta_{\overline{n}}\circ f\mid f\in [\overline{n}\to\overline{n}] \}$ is the set of extreme points of $\ST(n)$, this morphism is surjective. In fact this map is an isomorphism. Thus $\ST(n)$ is the free compact affine monoid over $[\overline{n} \to \overline{n}]$.
\end{example}

\subsection{A monad over compact groups}
Our final use of \discr\ to define a monad starts with \textsf{CGrp}, the category of compact groups and continuous group homomorphisms. Since \textsf{CGrp} is a subcategory of \textsf{CMon}, we know that applying \discr$_M$ to a compact group yields a compact affine monoid. However, $\discr_M(G)$ is not a group in general, so the forgetful functor from \textsf{CMon} does not take \discr$_M(G)$ to a group, but instead yields a compact monoid. 

But, when applied to a compact group $G$ \emph{qua} compact monoid, the unit of the monad \discr$_M$ sends each $g\in G$ to $\delta_g\in \disc{G}$, and this is a monoid --- hence  group --- homomorphism. So, we define a new functor $H\colon \textsf{CMon}\to \textsf{CGrp}$ by $H(S) = H(1_S)$, the group of units\footnote{A \emph{unit} of a monoid is an element that has a two-sided inverse with respect to the identity $1_S$. The set of units $H(1_S) = \{x\in S\mid (\exists y\in S)\ xy = yx = 1_S\}$ forms the largest subgroup of $S$ that has $1_S$ as the identity; if $S$ is compact, then so is the group of units.} of the compact monoid $S$. If $\phi\colon S\to T$ is a morphism of compact affine monoids, then $\phi\vert_{H(1_S)}\colon H(1_S)\to H(1_T)$ is a morphism of compact groups, so $H$ defines a functor. 
\begin{theorem} 
The functor $H\colon \textsf{CAM}\to \textsf{CGrp}$ is right adjoint to the functor \discr$_G\colon \textsf{CGrp}\to \textsf{CAM}$. In fact, the composition $H\circ \discr_G$ defines a monad on \textsf{CGrp}. Moreover, for any compact group $G$, we have $H(\discr_G{G}) \simeq G$. Again, the unit of the monad is the Dirac map, and the image of $G$ in $\discr_G(G)$ is the set of extreme points. 
\end{theorem}
\begin{example}
We again consider the case of classical channels. Here, given $n \geq 1$, $\ST(n)$ has for its group of units the permutation group $S(n)$. Applying \discr$_G$, we find that $\discr_G(S(n))$ is the free compact affine monoid over the group $S(n)$. But this is just the family $\DT(n)$ of doubly stochastic $n\times n$-matrices.
\end{example} 
We can use information about $\discr_G(S(n))$ to conclude information about $\DT(n)$. Wendel's Theorem~\cite{wendel} states that, for a compact group $G$, the compact monoid $\disc{G}$ has $\{ \delta_g\mid g\in G\}$ as its group of units, and the minimal ideal (every compact monoid has one --- cf.~\cite{hofmis}) is a zero, which in fact is Haar measure on $G$. In the case of $S(n)$, this reaffirms that the units of $\DT(n)$ are the permutations of $\overline{n}$, and that $\DT(n)$ has a zero, which is the equidistribution $\sum_{i\leq n} {1\over n} \delta_i$. 

\begin{corollary}
If $G$ is a finite group, then $\discr_G(G)$ is the free affine monoid over $G$, as well as the being the free compact affine monoid over $G$.
\end{corollary}
\begin{proof}
If $G$ is a finite group, then  $\discr(G) = \{ \sum_{i\leq k} r_i\delta_{g_i}\mid k\in\Nat, r_i\in [0,1], \sum_i r_i = 1\ \wedge\ g_i\in G\}$ consists of simple measures. If $S$ is an affine monoid and $\phi\colon G\to S$ is a monoid homomorphism, then $\phi(G)\subseteq H(1_S)$, and so $\phi\colon G\to H(1_S)$ is a group homomorphism. Then $\widehat{\phi}(\sum_{i\leq k} r_i \delta_{g_i}) = \sum_{i\leq k} r_i \phi(g_i)$ is easily seen to be a morphism of affine monoids that satisfies $\widehat{\phi}(\delta_g) = \phi(g)$ for each $g\in G$, and $\widehat{\phi}$ is the unique such since $\discr_G(G)$ consists of simple measures. This shows $\discr_G(G)$ is the free affine monoid over $G$, and the Theorem implies it also is the free compact affine monoid over $G$ since $G$ is finite, and hence compact. 
\end{proof}
\begin{remark}
In~\cite{crowd,MartinFeng}, the free affine monoid over a finite group is employed to deduce properties of quantum channels. The Corollary shows that the free affine monoid over $G$ is nothing other than $\discr_G(G)$, which implies it is compact, as well as telling us that it has a zero --- the uniform distribution on $G$. We believe other useful properties about quantum channels over finite groups can be deduced from this observation.
\end{remark}
\section{Capacity as a topological concept}\label{sec:top}
In this section we develop a new approach to understanding the capacity of a classical channel. Our idea is to analyze capacity from a topological perspective, rather from the usual perspective of inequalities prevalent in information theory. We begin with a brief reprise of  the basics of Shannon information; the standard reference for this material is \cite{cover}.

If $X\colon \X\to \Re$ is a random variable on a finite probability space $(\X,p)$, then the \emph{entropy}\footnote{We use $H(X)$ to denote the entropy of a random variable $X$; this overloads our notation for the maximal subgroups of a monoid $S$, but we believe the context will be sufficient to make the meaning clear.}  of $X$ is defined as $H(X) = -\sum_{x\in \X} p(x)\log_2p(x)$.  If $Y\colon \Y\to \Re$ is another finite random variable, then the \emph{conditional entropy} of $Y$ given $X$ is 
\begin{equation}\label{eqn:condentrop}
H(Y\vert X) = \sum_{x\in \X} p(x) H(Y| X=x) = \sum_{x\in\X}p(x)\sum_{y\in\Y} p(y|x) \log_2 {1\over p(y|x)},
\end{equation}
and the \emph{mutual information in $X$ and $Y$} is 
$$\I(Y,X) = H(Y) - H(Y|X) = H(X) - H(X|Y).$$
If $C\colon \X\to \Y$ is a channel from inputs $\X$ to outputs $\Y$, then $C$ is an $\X\times \Y$-matrix whose $(x,y)$-entry is the \emph{conditional probability}  of output $y$ occurring, given that the input  was $x$. Each distribution $p$ on the inputs $\X$ then produces a corresponding distribution $p\cdot C$ on $\Y$. The \emph{capacity of a channel} is given by
$$
\textsf{Cap}\colon [\discr(\X)\to\discr(\Y)]\to [0,1]\quad \mathrm{by}\quad 
\Capp{C} = \sup_{p\in Pr(\X)} H(p\cdot C) - H(p\cdot C \mid p),
$$
i.e., $\Capp{C}$ is the supremum of the possible mutual information values $\I(p\cdot C \mid p)$ as $p$ ranges over the distributions on $\X$, the set of inputs. 

If we let $\X = \Y = \overline{n}$ and $C\colon \overline{n}\to\disc{\overline{n}}$ is a channel, then  
\begin{equation}\label{eqn:1}
\textsf{Cap}(C) = \sup_{\sum_{i\leq n} r_i\delta_{i}}
\left[H(\sum_{j\leq n} r_jC(j|1),\ldots, \sum_{j\leq n} r_jC(j|n)) - 
\sum_{i\leq n} r_i H(C(i|1),\ldots,C(i|n))\right].
\end{equation}
This formula requires some interpretation. 
\begin{enumerate}
\item First, the term to which $H$ is first applied --- $(\sum_{j\leq n} r_jC(j,1),\ldots, \sum_{j\leq n} r_jC(j,n))$ ---  represents a distribution on $\Y = \overline{n}$ obtained from $pC$, where $p = \sum_i r_i \delta_i$ is a distribution on $\X = \overline{n}$. This is the $p$-convex combination of the $n$ vectors $(C(1,1),\ldots,C(1,n)), \ldots, (C(n,1),\ldots,C(n,n))\in [0,1]^{n}$ comprising the rows of the channel $C$, where $C(i,j)$ denotes the $i,j$-entry of $C$, interpreted as a conditional probability.  Since $C$ is a channel, each of these rows is a probability distribution on $\overline{n}$. (As a sanity check, we see that applying $H$ to the convex combination $(\sum_{j\leq n} r_jC(j,1),\ldots, \sum_{j\leq n} r_jC(j,n))$ thus makes sense, since a convex combination of probability distributions is another such, and $H$ applies to probability distributions.) 

\item Now, the convex combination $pC = (\sum_{j\leq n} r_jC(j,1),\ldots, \sum_{j\leq n} r_jC(j,n))$ is a point on the polytope $K\subseteq [0,1]^n$ the rows of $C$ generate, so $$(\sum_{j\leq n} r_jC(j,1),\ldots, \sum_{j\leq n} r_jC(j,n)), H(\sum_{j\leq n} r_jC(j,1),\ldots, \sum_{j\leq n} r_jC(j,n)))\in [0,1]^n\times \Re$$
represents the point on the surface $H$ generates over the polytope $K$. 

\item Likewise the second term, $\sum_{i\leq n} r_i H(C(i|1),\ldots,C(i|n))$ of Equation~\ref{eqn:1} is  a $p$-convex combination, $p = \sum_{i\leq n} r_i\delta_i$, of the terms $H(C(i,1),\ldots,C(i,n))$, each of which is obtained by applying $H$ to a row of $C$, regarded as an elemnt of $[0,1]^n$. We can regard each of the points $H(C(i,1),\ldots,C(i,n))$ as being the $n+1$-coordinate of a tuple $(C(i,1),\ldots,C(i,n), H(C(i,1),\ldots,C(i,n))\in\Re^{n+1}$, and hence the $p$-convex combination of these points lies on the polytope these points generate.

\item Finally, the difference $H(\sum_{j\leq n} r_jC(j|1),\ldots, \sum_{j\leq n} r_jC(j|n)) - 
\sum_{i\leq n} r_i H(C(i|1),\ldots,C(i|n))$ is the difference in the $n+1$-coordinates described under 2.\ and 3., so it is height of the vertical line between the point $pC$ in 3.\ and the corresponding point 
$$(C(1,1),\ldots,C(1,n), H(C(1,1),\ldots,C(1,n)),\ldots,(C(n,1),\ldots,C(n,n), H(C(n,1),\ldots,C(n,n))))$$
on the surface $H$ generates over $K$. 
\end{enumerate}
Thus, $\Capp{C}$ as presented by Equation~\ref{eqn:1} takes the supremum of the differences between the value of $H$ at a convex combination of the rows of $C$ and the same convex combination of $H$ applied to the rows of $C$. It is well-known that entropy $H$ is a strictly concave function, and we now take advantage of this to formulate a result about \textsf{Cap}.

\begin{definition}\label{def:strconc}
Let $K\subseteq \Re^n$ be a convex set. A function $f\colon K\to \Re$ is \emph{strictly concave} if 
$$f(r{\buildrel \to\over x} + (1-r){\buildrel \to\over y}) > rf({\buildrel \to\over x}) + (1-r)f({\buildrel \to\over y})$$
for all $r\in (0,1)$ and all ${\buildrel \to\over x}, {\buildrel \to\over y}\in K$.
\end{definition}
We next recall \emph{Jensen's Inequality:} 

\begin{theorem}[Jensen (cf.~\cite{hardy})]
If $f\colon K\to \Re$ is a convex function defined on a convex subset $K$ of a vector space $V$,  then $Ef(X)\geq f(E(X))$ for a finite random variable $X\colon \X\to K$, where $E$ denotes expectation. Moreover, if $f$ is strictly convex, then $E(f(X)) = f(E(X))$ implies $X$ is constant. 
\end{theorem}
Jensen's Inequality is a fundamental result of information theory; for example, it is crucial for proving mutual information is non-negative, that the mutual information in a pair of random variables is $0$ iff the random variables are independent, and that entropy itself is strictly concave (cf.~\cite{cover}, Chapter 2).  Since $f$ is (strictly) concave iff $-f$ is (strictly) convex, the  following  generalizes Jensen's Inequality by strengthening the result in case $f$ is strictly convex.

\begin{lemma}
If $f\colon K\to \Re$ be defined from a convex subset $K$ to $\Re$. Then the following are equivalent:
\begin{enumerate}
\item $f$ is strictly concave.
\item For all $r_1,\ldots, r_m\in (0,1)$ and all ${\buildrel \to\over x_1},\ldots , {\buildrel \to\over x_m}
\in K$, 
$$\sum _{i\leq m} r_i = 1\quad\Rightarrow\quad f\left (\sum_{i\leq m} r_i{\buildrel \to\over x_i}\right) >
\sum_{i\leq m} r_i f({\buildrel \to\over x_i}).$$
\end{enumerate}
\end{lemma}

\begin{proof}
$(ii)\Rightarrow (i)$ is obvious. For the reverse direction, we proceed by induction on $m$. The base case, $m = 2$, is just the definition of strict concavity. So suppose (ii) holds for some $m$, and consider a family $r_1,\ldots, r_{m+1}\in [0,1]$ and ${\buildrel \to\over x_1},\ldots , {\buildrel \to\over x_{m+1}} \in K$. Since $r_i\in (0,1)$ for each $i$,
\begin{eqnarray*}
f\left (\sum_{i\leq m+1} r_i{\buildrel \to\over x_i}\right) &=& 
f\left (\sum_{i\leq m-1} r_i{\buildrel \to\over x_i} 
+ (r_m + r_{m+1})({r_m\over r_m + r_{m+1}} {\buildrel\to\over x_m} 
+ {r_{m+1}\over r_m + r_{m+1}} {\buildrel\to\over x_{m+1}})\right ) \\
& > & \sum_{i\leq m-1} r_i f({\buildrel \to\over x_i}) + (r_m + r_{m+1}) 
f({r_m\over r_m + r_{m+1}} {\buildrel\to\over x_m} 
+ {r_{m+1}\over r_m + r_{m+1}} {\buildrel\to\over x_{m+1}})\\
&> & \sum_{i\leq m-1}  r_i f({\buildrel\to\over x_i}) + r_m f({\buildrel\to\over x_m}) + r_{m+1} f({\buildrel\to\over x_{m+1}})\\
& = & \sum_{i\leq m+1} r_i f({\buildrel\to\over x_i}).\\
\end{eqnarray*}

\vspace{-.5in}
\end{proof}\ 

\medbreak
\noindent\textbf{Notation:}
If $K\subseteq \Re^n$ be a compact convex set, then $\textsf{Con}_n(K)$ denotes the family of convex polytopes $\conv( \{x_1,\ldots,x_k\})$ generated by finite subsets $\{x_1,\ldots,x_k\}\subseteq K$, where $k\leq n$. 

Also note that $P_n \defi \{ x\in [0,1]^n\mid \sum_i x_i = 1\}$ is a compact, convex subset of $[0,1]^n$, which we identify with the family $\disc{\overline{n}}$ of probability distributions on $\overline{n}$.
\medbreak
\begin{proposition}\label{prop:strictconvex}
Let $K\subseteq \Re^n$ be a compact, convex set, and let  $f\colon K\to \Re$ be continuous and strictly concave. Define 
$$\widehat{f}\colon \textsf{Con}_n(K)\to \Re^{op}\quad \mathrm{by}\quad 
\widehat{f}(\conv(\{ x_1,\ldots,x_k\})) = \sup_{(r_1,\ldots,r_k)\in [0,1]^k} f\left (\sum_{i\leq k} r_i x_i\right ) - \sum_{i\leq k} r_i f(x_i).$$
Then $\widehat{f}$ is continuous and  monotone with respect to reverse inclusion.
\end{proposition}
\begin{proof}
The compactness of $K$ implies that the family $\textsf{Con}_n(K)$ is closed under filtered intersections in the hyperspace of non-empty, closed subsets of $K$, and then the continuity of $\widehat{f}$ follows from the continuity of $f$. This map is clearly monotone. To show it is strictly monotone, let $\conv(\{ x_1,\ldots, x_k\}), \conv( \{y_1,\ldots, y_m\})\in \textsf{Con}_n(K)$ with 
$\conv(\{x_1,\ldots,x_k\}) \subsetneq \conv( \{y_1,\ldots, y_m\})$. Since $f$ is continuous, $\widehat{f}(\conv(\{ x_1,\ldots,x_k\}))$ assumes its value at some point in $\conv(\{ x_1,\ldots,x_k\})$, and since $f$ is strictly concave, this value is not assumed at $x_i$ for any index $i$. Thus, there is a $k$-tuple $(r_1,\ldots,r_k)\in (0,1)^k$ with 
$$\widehat{f}(\conv(\{ x_1,\ldots,x_k\})) = f\left (\sum_{i\leq k} r_i x_i\right ) - \sum_{i\leq k} r_i f(x_i).$$ Because $\conv(\{ x_1,\ldots,x_k\})\subsetneq \conv(\{ y_1,\ldots,y_n\})$, for each $i\leq k$, there is $(s_{i,1},\ldots, s_{i,m})\in [0,1]^m$ with $x_i = \sum_{j\leq m} s_{i,j} y_j$, and at least one of the families $(s_{i,j})_{j\leq m}\in (0,1)^m$. For this index $i$, we have 
$f(x_i) = f(\sum_{j\leq m} s_{i,j} y_j) > \sum_{j\leq m} s_{i,j} f(y_j)$. Lemma~\ref{def:strconc}(ii) then implies that 
$$\sum_{i\leq k} r_i f(x_i) = \sum_{i\leq k} r_if\left (\sum_{j\leq m} s_{i,j} y_j\right ) > \sum_{i\leq k}\sum_{j\leq m} r_i s_{i,j} f(y_j),$$ and so
\begin{eqnarray*}
\widehat{f}(\conv(\{ x_1,\ldots,x_k\})) &=& f(\sum_{i\leq k} r_i x_i) - \sum_{i\leq k} r_i f(x_i)\\
&\leq& f\left (\sum_{i\leq k} r_i \left (\sum_{j\leq m} s_{i,j} y_j\right)\right) - \sum_{\i\leq k}\sum_{j\leq m} r_i s_{i,j} f(y_j)\\
&\leq& \sup_{(s_1,\ldots, s_m)\in [0,1]^m} f\left (\sum_j s_j y_j\right ) - \sum_j s_j f(y_j)\\ 
&= & \widehat{f}(\conv(\{ y_1,\ldots , y_m\})).
\end{eqnarray*}

\vspace{-.35in}
\end{proof}
\section{Domains}
In this section, we introduce domains, which are the next ingredient in our analysis of classical channels. For details about these structures, a standard reference is \cite{abrjung} or \cite{comp}. A \emph{partial order} is a non-empty set $P$ endowed with a reflexive, antisymmetric and transitive relation. A subset  $D\subseteq P$ is \emph{directed} is every finite subset of $D$ has an upper bound in $D$; $P$ is \emph{directed complete} if every directed subset of $P$ has a least upper bound in $P$. We denote directed complete partial orders as dcpos. 

If $P$ and $Q$ are dcpos, then $f\colon P\to Q$ is \emph{Scott continuous} if $f$ is monotone and preserves suprema of directed sets. An equivalent definition is available using topology: a subset $U\subseteq P$ is \emph{Scott open} if $U = \ua U = \{ x\in P\mid (\exists u\in U)\ u\leq x\}$ is an upper set, and for any directed subset $D\subseteq P$, if $\sup D\in U$, then $D\cap U\not=\emptyset$. The Scott-open sets form a topology on $P$, called the Scott topology, and the functions $f\colon P\to Q$ that are continuous with respect to this topology are exactly those that are Scott continuous, as defined above. 

\begin{example}\label{exam:dom}
Let $K$ be a compact convex subset of a topological vector space, and let $Con(K)$ denote the compact convex subsets of $K$. We can order these by reverse inclusion: $C\sqsubseteq C'\ \Leftrightarrow\ C'\subseteq C$. A directed family $D\subseteq Con(K)$ is simply a filterbasis, and since $K$ is compact and each set in $D$ is convex, the set $\bigcap D\in Con(K)$. Thus $Con(K)$ is a dcpo. 

We can say more. If $C, C'\in Con(K)$ and $C'\subseteq C^\circ$, the interior of $C$, then given any directed set $D$ with $\bigcap D\subseteq C'$, there is some $E\in D$ with $E\subseteq C^\circ$, and hence $E\subseteq C$. In this case we say $C$ \emph{is way-below} $C'$, and we write $C\ll C'$. In fact, if the ambient topological vector space is locally convex, then each $C'\in Con(K)$ is the filtered intersection of those $C$ satisfying $C\ll C'$: this follows from the fact that in any compact Hausdorff space, each compact subset is the filtered intersection of its compact neighborhoods, and the same applies to compact, convex sets in a locally convex topological vector space.
\end{example} 
A \emph{domain} is a dcpo satisfying $\{ y\in P\mid y\ll x\}$ is directed and $x = \sup \{y\in P\mid y\ll x\}$ for each $x\in P$. The original motivation for domains was to provide semantic models for high-level programming languages, where the fact that any Scott-continuous selfmap on a domain with least element has a least fixed point gives a canonical model for recursion. 

Motivated by examples of selfmaps that are not Scott continuous, Martin~\cite{martin:thesis} devised another approach to guaranteeing fixed points, using the concept of a \emph{measurement}: A  Scott-continuous function $m\colon P\to [0,\infty)^{op}$  is said to \emph{measure the content at $x\in P$} if, given $U\subseteq P$ Scott open, $$x\in U\ \Rightarrow\ (\exists \epsilon > 0)\ m_\epsilon(x)\subseteq U,$$
 where $m_\epsilon(x) = \{ y\leq x\mid m(y) - m(x) < \epsilon\}$. We say that $m$ \emph{measures $P$} if $m$ measures the content at $x$ for each $x\in P$. 

For our next result, we need some notions from topology. Recall that s subset $A\subseteq X$ of a topological space is \emph{saturated} if $A = \bigcap \{ U\mid A\subseteq U\ \mathrm{open}\}$. The \emph{saturation} of a subset $A$ is $\bigcap \{ U\mid A\subseteq U\ \mathrm{open}\}$, so a set is saturated iff it is equal to its saturation. Moreover, a subset is compact iff its saturation is compact. 
 \begin{definition}
 A continuous function $f\colon X\to Y$ between topological spaces is \emph{proper} if $f^{-1}(K)$ is compact for each saturated, compact subset $K\subseteq Y$. 
 \end{definition}
A Scott-continuous map $m\colon P\to Q$ between continuous posets is proper iff $m^{-1}(\ua y)$ is Scott compact for each $y\in Q$.  We say that a Scott-continuous mapping $m\colon P\to Q$ between posets is \emph{proper at $x\in P$} if $\da x \cap m^{-1}(\ua y)$ is Scott compact in $\da x$ for each $y\in Q$. We use $[0,\infty)^{op}$ to denote the non-negative real numbers in the \emph{dual order}; the following result is from \cite{hofmis}:
 \begin{proposition}\label{prop:strictrmono}
Let $P$ be a domain and let $m\colon P\to [0,\infty)^{op}$ be Scott continuous. If $m$ is proper at $x\in P$, then the following are equivalent:
\begin{enumerate}
\item $m$ measures the content at $x$.
\item  $m$ is strictly monotone at $x\in P$: i.e., $y\leq x\ \&\ m(y) = m(x)\ \Rightarrow\ y=x.$
\end{enumerate}
In particular, a Scott-continuous, proper map $m\colon P\to [0,\infty)^{op}$ measures $P$ iff $m$ is strictly monotone at each $x\in P$. 
\end{proposition}

\begin{corollary}\label{prop:strictconv}
Let $P_n \defi \{ x\in [0,1]^n\mid \sum_i x_i = 1\}$ be the compact, convex set of distributions on $\overline{n}$. Then $(\textsf{Con}_n(P_n),\supseteq)$ is a domain, and  the mapping 
$\textsf{cap}\colon \textsf{Con}_n(P_n)\to [0,\infty)^{op}$ by 
$$\textsf{cap}(\conv( F)) = \sup\, \left\{H(\sum_{x\in F} r_x\cdot x) - (\sum_{x\in F} r_xH(x)) \mid r_x\geq 0, \sum_{x\in F} r_x = 1\right\}$$
measures $(\textsf{Con}_n(P_n),\supseteq)$.
\end{corollary}
\begin{proof} The discussion in Example~\ref{exam:dom} applies to $K = P_n$ to show that $\textsf{Con}(P_n)$ is a domain, and $\textsf{Con}_n(P_n)$ is closed in $\textsf{Con}(P_n)$ under filtered intersections. Since $\Re^n$ is locally convex, it's  easy to show that each $\conv(F)$ is the intersection of sets $\conv(G)$, where $\conv(F)\subseteq \conv(G)^\circ$ and $|G|\leq n$ if $|F|\leq n$. Hence $(\textsf{Con}_n(P_n),\supseteq)$ is a domain. 

A compact, saturated subset of $[0,\infty)^{op}$ has the form $A = [0,r]$ for some $r$, and then 
$\textsf{cap}^{-1}([0,r]) = \{ \conv(F)\mid \textsf{cap}(\conv(F) \leq r\}$ is closed, and hence compact, since $\textsf{cap}$ is continuous and $\textsf{Con}(P_n)$ --- and hence also $\textsf{Con}_n(P_n)$ --- are compact in the Lawson topology (cf.~\cite{abrjung,comp}). But $\textsf{cap}^{-1}([0,r])$ also is an upper set, so it is Scott compact. Thus \textsf{cap} is proper, and so it measures \textsf{Con}$_n(P_n)$ iff it is strictly monotone. But the latter follows from Proposition~\ref{prop:strictconvex}, since entropy is strictly concave.
\end{proof}
\noindent\textbf{Remark:}
As we will see in a moment, the real import of this last result is not so much that \textsf{cap} measures $\textsf{Con}_n(P_n)$, \emph{per se}, but rather that this implies the mapping \textsf{cap} is strictly monotone. This will tell us that under an appropriate (pre-)order, having one channel strictly below another implies that the capacity of the lower channel is strictly less than that of the larger one. 

We also note that Martin and Panangaden have obtained results in \cite{panangad} that can be used to derive the Proposition~\ref{prop:strictrmono} and Corollary~\ref{prop:strictconv}.

\section{A domain-like structure of $\ST(n)$}
We have seen that $\ST(n)$ is a compact affine monoid, and we already commented that every compact semigroup has a unique smallest ideal: i.e., a non-empty subset $I\subseteq S$ satisfying $IS \cup SI \subseteq I$. This \emph{minimal ideal} is denoted $\M(S)$, and it is closed, hence compact. For example, we noted that $\M(\DT(n))$ is a point, which is the equidistribution on $\overline{n}$. A reference for much of the material in this section is \cite{hofmis}, where basic results about compact affine monoids are laid out. A good reference for results about transformation semigroups can be found in \cite{cliff}.
 
\begin{proposition}\label{prop:affinemonoid}
If $C\in \ST(n)$, then $C(\disc{\overline{n}}) = \conv(\{\delta_i\mid 1\in\overline{n}\})$ is a convex polytope in $[0,1]^n$.
\end{proposition}
\begin{proof} 
If $C\in \ST(n)$, then $C\colon \disc{\overline{n}}\to \disc{\overline{n}}$ is an affine mapping, so it preserves the convex structure of $\disc{\overline{n}}$. It follows that $C(\disc{\overline{n}}) = \conv(\{\delta_iC\mid i\leq n\})$, where $\delta_i\in\disc{\overline{n}}$ is the Dirac measure on $i\in \overline{n}$. Thus, $C(\disc{\overline{n}})$ is a convex polytope in $\disc{\overline{n}}$. 
\end{proof}

As a result of the Proposition, we can define a relation on $\ST(n)$ by 
\begin{equation}\label{eqn:equiv1}
C \equiv C'\ \Leftrightarrow\ C(\disc{\overline{n}}) = C'(\disc{\overline{n}}).
\end{equation} 
This is clearly a closed equivalence relation, and because channels are affine maps. 
\begin{equation}\label{eqn:equiv2}
C \equiv C'\ \Leftrightarrow\ C(\{\delta_i\mid i\in \overline{n}\}) = C'(\{\delta_i\mid i\in \overline{n}\})
\end{equation}
Now, $C(\delta_i) = \delta_iC = C(i)$, the i$^{th}$ row of $C$, so $C \equiv C'$ iff $C$ and $C'$ have the same set of rows vectors. Hence, 
\begin{equation}\label{eqn:equiv3}
C \equiv C'\ \Leftrightarrow\ (\exists \pi\in S(n))\ M_\pi C = C',
\end{equation}
where $M_\pi$ is the stochastic matrix representing the permutation $\pi\in S(n)$. 

We also can obtain an algebraic representation of the relation $\equiv$ using the monoid structure of $\ST(n)$.
 \begin{definition}
 If $X$ is a set, then the \emph{full transformation semigroup $T(X)$ on $X$} is the family of all selfmaps of $X$ under composition. A \emph{transformation semigroup} is a subsemigroup of $T(X)$ for some set $X$.  
 \end{definition}
 
 \noindent\textbf{Notation:} If $S\subseteq T(X)$ is a transformation semigroup, then for $s,s'\in T(X)$ and $x\in X$, the element $ss'\in S$ denotes the function $ss'(x) = s'(s(x))$ --- i.e., we use the ``algebraic notation" for function application, which agrees with our representation matrix multiplication as composition of functions. 
\medbreak

 Here are some simple observations about $T(X)$; the proofs are all straightforward: 
 \begin{enumerate}
 \item $T(X)$ is a monoid whose group of units is the family of bijections of $X$; if $X$ is finite, this is just $S(|X|)$, the group of permutations of $|X|$-many letters. 
 \item Each constant map $f_x\colon X\to X$ by $f_x(y) = x$ is s \emph{left zero} in $T(X)$: if $g\in T(X)$, then $gf_x = f_x$. It follows from general semigroup theory that $\M(T(X)) = \{ f_x\mid x\in X\}$. 
 \item If $S$ is a transformation semigroup on $X$ and $S\cap \M(T(X)) \not=\emptyset$, then $\M(S) = S\cap \M(T(X))$. Thus, if $S$ contains a constant map, then $\M(S)$ consists of constant maps. This follows from the fact that $gf $ is a constant map if either $f$ or $g$ is one, so the constant maps in $S$ form an ideal. 
 \item If $S$ is a transformation semigroup, then for each $s\in S$ and each $x\in X$, $f_x\in S^1s\ \Rightarrow\ x\in s(X)$.\footnote{If $S$ is a semigroup, then $S^1$ denotes the semigroup $S$ with an identity element adjoined.}  Indeed, if $f_x\in Sx\cup \{s\}$, then there is some $s'\in S$ with $f_x = s's$, so $x = s(s'(y))  \in S(X)$. 
 \item Conversely, if $M(S) = \{ f_x\mid x\in X\}$, then $x\in s(X)\ \Rightarrow\ f_x\in S^!s$. This follows since $x\in s(X)$ implies $x = s(y)$, for some $y\in X$; if $s\not=f_x$, then $f_x = f_ys\in Ss$.
 \end{enumerate}
 
 \begin{definition}
 Let $S$ be a monoid. We define the relations $\equiv_M$ and $\leq_M$ on $S$ by:
 \begin{eqnarray}\label{eqn:2}
 s\equiv_M s'\ &\Leftrightarrow&\ s\M(S) = s'\M(S),\ \mathrm{and}\\
 s \leq_M s'\ &\Leftrightarrow&\ s\M(S)\subseteq s'\M(S).\notag
 \end{eqnarray}
 \end{definition}
 It is routine to show that $\equiv_M$ and $\leq_M$ are both (topologically) closed relations on any compact monoid $S$. 
 
 Combining properties 4.\ and 5.\ above on transformation semigroups with equivalences~\ref{eqn:equiv1}---\ref{eqn:equiv3} yields:
 
 \begin{proposition}\label{prop:mequiv} Let $n\geq 1$ and let $C,C'\in\ST(n)$. Then
\begin{equation}\label{eqn:equiv4}
C(\disc{\overline{n}}) = C'(\disc{\overline{n}})\ \Leftrightarrow\ \M(\ST(n))C = \M(\ST(n))C'\ 
 \Leftrightarrow\ (\exists \pi\in S(n))\ M_\pi C = C'.
 \end{equation}
 \end{proposition} 
 Since $\equiv_M = \leq_M\cap\, (\leq_M)^{-1}$, we can form the relation $\leq_M/\!\equiv_M$, which is a closed partial order on $\ST(n)/\!\equiv_M$. 
 
 \begin{theorem}\label{thm:main} 
 Let $n\geq 1$. Then
 \begin{enumerate}
 \item  The relation $\equiv_M$ is a left congruence on $\ST(n)$. 
 \item $(\ST(n)/\!\equiv_M,\leq_M/\!\equiv_M)$ is a compact ordered space, and the 
 quotient map $\pi\colon \ST(n)\to \ST(n)/\!\equiv_M$ is an monotone map.  
 \item As an ordered space, $\ST(n)/\!\equiv_M\ \simeq\ \textsf{Con}_n(P_n)$. 
 \item Thus $\textsf{cap}\colon (\ST(n)/\!\equiv_M,\leq_m/\!\equiv_M)\to [0,\infty)^{op}$ measures $(\ST(n)/\!\equiv_M,\leq_m/\!\equiv_M)$.
 \item If $C\in \ST(n)$, then $\textsf{Cap}(C) = \textsf{cap}(\conv(C(1),\ldots, C(n)))$, where $C(i)$ is the $i^{\mathrm{th}}$ row of $C$.
 \end{enumerate}
 \end{theorem}
 \begin{proof}
 For 1., it is clear from Equation~\ref{eqn:2} that $C \equiv_M C'$ implies $C''C\equiv_M C''C'$. 
 Thus, $\equiv_M$ is a left-congruence. 
 
Because $\equiv_M$ is a closed relation and $\ST(n)$ is compact, the quotient space is compact and the quotient map is closed and continuous. 
The definition of the pre-order $\leq_M$ and the quotient order $\leq_M/\!\equiv_M$ implies the quotient map is monotone.
 
 3.\ follows from Proposition~\ref{prop:mequiv}, from which 4.\ and 5.\ are clear. 
 \end{proof}
 
 So we see that $\ST(n))$ has a natural pre-order $\leq_M$ defined by its algebraic structure as a compact monoid, and if $C\leq_M C'$ then $\textsf{Cap}(C)\leq \textsf{Cap}(C')$. Moreover, this pre-order turns into a partial order on $\ST(n)/\!\equiv_M$, and here capacity is the mapping \textsf{cap}. Importantly, \textsf{cap} --- and hence \textsf{Cap} on $\ST(n)$ --- is strictly monotone with respect to this partial (pre-) order. Moreover, $C \equiv_M C'$ iff $S(n)C = S(n) C'$, so each is a permutation of the rows of the other. 

\begin{example}
Here's an example of what  our results tell us about $\ST(n)$. Recall that a \emph{Z-channel} is a binary channel of the form $$Z_p = \left(\begin{matrix} 1-p &p\\ 0 & 1\end{matrix}\right)= (1-p)\left ( \begin{array}{cc}
	1 & 0 \\
	 0 & 1\end{array}\right ) + 
p \left ( \begin{array}{cc}
	0 & 1 \\
	 0 & 1\end{array}\right ),$$
showing that each lies on a one-parameter semigroup $\psi\colon ([0,1], \cdot)\to \ST(2)$ by $\psi(p) = Z_p$. Now, $\psi$ is a homomorphism, so $p < p'\ \Rightarrow Z_p = Z_{p \over p'} Z_{p'}$, while obviously, $p \not= p'\ \Rightarrow\ S(2) Z_p \cap S(2) Z_{p'} = \emptyset$. It follows that $p < p'\ \Rightarrow\ \pi(Z_p) < \pi(Z_{p'})\ \Rightarrow\ \Capp{Z_p} = \textsf{cap}(\pi(Z_p)) < 
\textsf{cap}(\pi(Z_{p'})) = \Capp{Z_{p'}}$, so the Z-channels $\{Z_p\mid 0\leq p \leq 1\}$ all have distinct capacities. 

Similarly, the matrices 
$$Z'_p = \left(\begin{matrix} 1 & 0\\ 1-p & p\end{matrix}\right)
= p\left ( \begin{array}{cc}
	1 & 0 \\
	 0 & 1\end{array}\right ) + 
(1-p) \left ( \begin{array}{cc}
	1 & 0 \\
	 1 & 0\end{array}\right ),$$
form a one parameter semigroup from $I_2$ to $\M(\ST(2))$, along which \textsf{Cap} is strictly decreasing. Now $\M(\ST(2)) = \{r\cdot O_1 + (1-r) O_2\mid 0\leq r \leq 1\}$, where $O_i$ is the matrix both of whose rows are $(\delta_{1i}\ \delta_{2i})$, for $i = 1,2$.  For each fixed $r\in [0,1]$, there is a one-parameter semigroup $p\mapsto p\cdot I_2 + (1-p)\cdot (r\cdot Z_p + (1-r)\cdot Z'_p)$, and combining the earlier results, we conclude that along this one-parameter semigroup, \textsf{Cap} is strictly decreasing. Note as well that $\conv(\{I_2\}\cup \M(ST(2)))$ is equal to the union of these one-parameter semigroups. 

We can generalize this \emph{verbatim} to $\ST(n)$: define a Z-channel in $\ST(n)$ to be one of the form $Z_p  = p\cdot I_n + (1-p)\cdot O_k$, where $p\in [0,1]$ and $O_k$ is the channel in $\M(\ST(n))$ all of whose rows are $(0,\ldots, 0, 1,0 ,\ldots)$, where the unique $1$ appears in the $k^{th}$ entry. As in the binary case, $p \leq p'\ \Rightarrow Z_p = Z_{p\over p'}Z_{p'}$, and $p\not= p'\ \Rightarrow\ S(n) Z_p\cap S(n) Z_{p'} = \emptyset$. So we can again conclude that \textsf{Cap} is strictly decreasing along this one-parameter semigroup. As in the case of $n=2$, $\M(\ST(n)) = \{ \sum_{i\leq n} r_k O_k\mid \sum_k r_k = 1\}$, so the result extends \emph{verbatim} to these one-parameter semigroups. 
\end{example}
 \section{Summary and future work}
We have used an array of tools to analyze the capacity map on the set of classical channels. In Section~\ref{sec:top}, we gave a topological interpretation of capacity of a channel: it is the maximum distance between the surface generated by the entropy function applied to the rows of the channel, and the polytope generated by the entropy function applied to each individual row. This suggests an method for computing capacity: the Generalized Mean Value Theorem implies the  capacity is achieved at the unique place where gradient of the capacity function is $0$ --- this point is unique because entropy is strongly concave. Moreover, this produces the input distribution where capacity is achieved --- the celebrated Arimoto--Blahut Algorithm~\cite{arim,blah} commonly used to compute the capacity of a discrete memoryless channel is an iterative procedure that approximates the capacity, not the input distribution where its value is assumed. An algorithm  more closely related to our results can be found in \cite{wata}, where the iteration scheme follows the concavity of the capacity function using Newton's Method. 

We also applied our topological result to derive a domain-theoretic interpretation of capacity, again using the strong concavity of entropy: the family $Con_n(P_n)$ of polytopes with at most $n$ vertices is a domain, and capacity measures this domain. The important point is that this implies capacity is strictly monotone with respect to the partial order. We found that $\ST(n)$ has a natural, algebraically-defined pre-order whose associated equivalence $\equiv_M$ defines a closed congruence on $\ST(n)$, and modulo which we obtain a copy of $\textsf{Conv}_n(P_n)$. This implies that capacity is strictly monotone with respect to the pre-order on $\ST(n)$. 

In addition, we used the probabilistic measures on compact spaces to define three monads, each of which tells us something about classical channels. The first realizes $\ST(m,n)$ as the morphisms on the Kleisli category of the monad \discr$_S$. The second shows that $\ST(n)$ is the free compact, affine monoid $\discr_M{[\overline{n}\to\overline{n}]}$, while the third shows that $\DT(n)$ is the free compact monoid over $S(n)$, the group of permutations on $n$ letters. This last also implies $\DT(n)$ has a zero, which is Haar measure on $S(n)$. 

The work discussed here concerns classical channels, but we believe that much of it could be generalized to the quantum setting. We pointed out one connection to existing work on the roll of free affine monoids in analyzing quantum qubit channels. In any case, we have shown that the basic results of \cite{martin2} generalize from the binary case, where the capacity function of a binary channel was studied in terms of the subinterval of $[0,1]$ it determines. The analog here is the polytope the rows of an $n\times n$-stochastic channel determine. 

\section*{Acknowledgment}
Many of the results in this paper are based on results in the paper \cite{hofmis}, and I want to express my thanks to my longtime colleague, \textsc{Karl H.\ Hofmann} for many enlightening and stimulating conversations about topics related to this work. 

\bibliographystyle{eptcs}

\end{document}